\documentclass[11pt,a4paper]{amsart}
\usepackage[margin=20mm]{geometry}
\usepackage[numbers]{natbib}
\usepackage{hyperref}

\usepackage{amsaddr}

\usepackage{amssymb,amsmath}
\usepackage{amsthm}
\usepackage{bbm}
\usepackage{stmaryrd}
\usepackage[usenames,dvipsnames]{xcolor}
\usepackage{color}
\usepackage{enumitem}

\input xy
\xyoption{all} 

\newtheorem{theorem}{Theorem}
\newtheorem{corollary}{Corollary}

\theoremstyle{definition}
\newtheorem{definition}{Definition}

\newtheorem{example}{Example}

\newcommand{\p}{\mbox{\boldmath$\rho$}}

\newcommand{\rme}{\textnormal{e}}

\DeclareMathOperator{\Vect}{Vect}

\DeclareMathOperator{\Hom}{Hom}

\DeclareMathOperator{\Div}{Div}

\font\black=cmbx10 \font\sblack=cmbx7 \font\ssblack=cmbx5 \font\blackital=cmmib10  \skewchar\blackital='177
\font\sblackital=cmmib7 \skewchar\sblackital='177 \font\ssblackital=cmmib5 \skewchar\ssblackital='177
\font\sanss=cmss10 \font\ssanss=cmss8 
\font\sssanss=cmss8 scaled 600 \font\blackboard=msbm10 \font\sblackboard=msbm7 \font\ssblackboard=msbm5
\font\caligr=eusm10 \font\scaligr=eusm7 \font\sscaligr=eusm5  \font\fraktur=eufm10
\font\sfraktur=eufm7 \font\ssfraktur=eufm5 
\font\bsymb=cmsy10 scaled\magstep2
\def\all#1{\setbox0=\hbox{\lower1.5pt\hbox{\bsymb
       \char"38}}\setbox1=\hbox{$_{#1}$} \box0\lower2pt\box1\;}
\def\exi#1{\setbox0=\hbox{\lower1.5pt\hbox{\bsymb \char"39}}
       \setbox1=\hbox{$_{#1}$} \box0\lower2pt\box1\;}

\def\tx#1{{\fam0\relax#1}}

\newfam\bifam
\textfont\bifam=\blackital \scriptfont\bifam=\sblackital \scriptscriptfont\bifam=\ssblackital

\newfam\blfam
\textfont\blfam=\black \scriptfont\blfam=\sblack \scriptscriptfont\blfam=\ssblack

\newfam\bbfam
\textfont\bbfam=\blackboard \scriptfont\bbfam=\sblackboard \scriptscriptfont\bbfam=\ssblackboard

\newfam\ssfam
\textfont\ssfam=\sanss \scriptfont\ssfam=\ssanss \scriptscriptfont\ssfam=\sssanss
\def\sss#1{{\fam\ssfam\relax#1}}

\newfam\clfam
\textfont\clfam=\caligr \scriptfont\clfam=\scaligr \scriptscriptfont\clfam=\sscaligr

\newfam\frfam
\textfont\frfam=\fraktur \scriptfont\frfam=\sfraktur \scriptscriptfont\frfam=\ssfraktur

\def\hpb#1{\setbox0=\hbox{${#1}$}
    \copy0 \kern-\wd0 \kern.2pt \box0}
\def\vpb#1{\setbox0=\hbox{${#1}$}
    \copy0 \kern-\wd0 \raise.08pt \box0}

\def\pmb#1{\setbox0\hbox{${#1}$} \copy0 \kern-\wd0 \kern.2pt \box0}
\def\pmbb#1{\setbox0\hbox{${#1}$} \copy0 \kern-\wd0
      \kern.2pt \copy0 \kern-\wd0 \kern.2pt \box0}
\def\pmbbb#1{\setbox0\hbox{${#1}$} \copy0 \kern-\wd0
      \kern.2pt \copy0 \kern-\wd0 \kern.2pt
    \copy0 \kern-\wd0 \kern.2pt \box0}
\def\pmxb#1{\setbox0\hbox{${#1}$} \copy0 \kern-\wd0
      \kern.2pt \copy0 \kern-\wd0 \kern.2pt
      \copy0 \kern-\wd0 \kern.2pt \copy0 \kern-\wd0 \kern.2pt \box0}
\def\pmxbb#1{\setbox0\hbox{${#1}$} \copy0 \kern-\wd0 \kern.2pt
      \copy0 \kern-\wd0 \kern.2pt
      \copy0 \kern-\wd0 \kern.2pt \copy0 \kern-\wd0 \kern.2pt
      \copy0 \kern-\wd0 \kern.2pt \box0}

\mathchardef\za="710B  
\mathchardef\zb="710C  
\mathchardef\zg="710D  
\mathchardef\zd="710E  
\mathchardef\zve="710F 
\mathchardef\zz="7110  
\mathchardef\zh="7111  
\mathchardef\zvy="7112 
\mathchardef\zi="7113  
\mathchardef\zk="7114  
\mathchardef\zl="7115  
\mathchardef\zm="7116  
\mathchardef\zn="7117  
\mathchardef\zx="7118  
\mathchardef\zp="7119  
\mathchardef\zr="711A  
\mathchardef\zs="711B  
\mathchardef\zt="711C  
\mathchardef\zu="711D  
\mathchardef\zvf="711E 
\mathchardef\zq="711F  
\mathchardef\zc="7120  
\mathchardef\zw="7121  
\mathchardef\ze="7122  
\mathchardef\zy="7123  
\mathchardef\zf="7124  
\mathchardef\zvr="7125 
\mathchardef\zvs="7126 
\mathchardef\zf="7127  
\mathchardef\zG="7000  
\mathchardef\zD="7001  
\mathchardef\zY="7002  
\mathchardef\zL="7003  
\mathchardef\zX="7004  
\mathchardef\zP="7005  
\mathchardef\zS="7006  
\mathchardef\zU="7007  
\mathchardef\zF="7008  
\mathchardef\zW="700A  
\mathchardef\zC="7009  

\newcommand{\be}{\begin{equation}}
\newcommand{\ee}{\end{equation}}

\newcommand{\bea}{\begin{eqnarray}}
\newcommand{\eea}{\end{eqnarray}}
\def\*{{\textstyle *}}
\newcommand{\R}{{\mathbb R}}

\newcommand{\Z}{{\mathbb Z}}

\newcommand{\s}{{\textstyle *}}

\def\Hom{\sss{Hom}}

\def\ul{\underline}

\def\Vect{\sss{Vect}}

\def\xi{\tx{i}}

\def\s*{{\scriptstyle *}}

\def\cO{\mathcal{O}}

\def\ul{\underline}

\newcommand{\beas}{\begin{eqnarray*}}
\newcommand{\eeas}{\end{eqnarray*}}

\def\half{\frac{1}{2}}

\title{Modular Classes and Supersymmetric Berezin Volumes}
\author{Andrew James Bruce}  
   \email{andrewjamesbruce@googlemail.com}
  
   \date{\today}
 
\begin{document}
\begin{abstract}
We argue that modular classes of Q-manifolds provide an efficient method for addressing the existence of supersymmetric Berezin volumes in the supergeometric representation theory of the $\mathcal{N}=2$ $d=1$ supertranslation algebra. We establish a cohomological coherence criterion for the existence of a Berezin volume that is invariant under both of the supercharges. \par
\end{abstract}

 \maketitle

%
%
In this short note, we highlight the application of the cohomology and modular classes of $Q$-manifolds  to the construction of supersymmetric Berezin volumes. Superspace methods provide a powerful geometric framework for constructing supersymmetric theories by generalising the notion of an action. A salient point is the behaviour of the Berezin volume/integration measure under supersymmetry transformations.  The standard constructions in superspace require the Berezin volume to be invariant under the action of the supercharges; checking this may be non-trivial.   For simplicity we concentrate on  the $\mathcal{N}=2$ $d=1$ \emph{supertranslation algebra} (see Nicolai \cite[Section 2]{Nicolai:1976} and Witten \cite[Section 3]{Witten:1982})
\begin{equation}\label{eqn:SupLieAlg}
    Q_1^2 = Q_2^2 =0\,, \qquad \{Q_1, Q_2\} = P\,, \qquad [P, Q_1] =0\,, \qquad [P,Q_2]=0\,.
\end{equation}
Note that via defining $Q_\pm := Q_1 \pm  Q_2$, the algebra \eqref{eqn:SupLieAlg} can be cast in the form
\begin{equation}\label{eqn:SupLieAlg2}
    Q^2_\pm = \pm P\,, \qquad \{Q_+, Q_-\} = 0\,, \qquad [P, Q_\pm] =0 \,,
\end{equation}
which is the $\mathcal{N}=(1,1)$ $d=1$ supertranslation algebra. These superalgebras are essential in supersymmetric classical and quantum mechanics; for an overview, the reader may consult Junker \cite{Junker:2019}. For more on the irreducible representations of $d=1$ extended superalgebras, the reader can consult \cite{Baulieu:2015,Pashnev:2001}.  We will concentrate on \eqref{eqn:SupLieAlg} in this letter and comment on \eqref{eqn:SupLieAlg2} in due course.\par
We study real and finite-dimensional supermanifolds $M = (|M|, \cO_M)$, equipped with a triple of vector fields $Q_1, Q_2 \in \Vect(M)$ odd and $P \in \Vect(M)$ even, that under the standard $\Z_2$-graded Lie bracket, satisfy the Lie superalgebra \eqref{eqn:SupLieAlg}. We set $C^\infty(M):= \cO_M(|M|)$ and speak of functions. As we have a pair of homological vector fields, $Q_i$ ($i = 1,2$), i.e., they `square to zero', we have two homological structures on $M$ (see Schwarz \cite[Section 3]{Schwarz:1993}). As there are two supercharges, there are two modular classes that control the existence of an invariant Berezin volume. \par
The first systematic studies of characteristic classes of Q-manifolds, including the modular class, were by Lyakhovich \& Sharapov \cite{Lyakhovich:2004} and Lyakhovich, Mosman, \& Sharapov \cite{Lyakhovich:2008,Lyakhovich:2010}. The author reviewed modular classes of Q-manifolds \cite{Bruce:2017}, which we will draw heavily from. Modular classes were first defined for Poisson manifolds and were subsequency extenced to Lie algebroids; see the review  article of  Kosmann-Schwarzbach \cite{Kosmann-Schwarzbach:2008} for more details.
 \par  
Recall that the divergence of a vector field $X$ is given defined by $\p \,\Div_{\p} X:=  L_X \p$, where $\p$ is a chosen Berezin volume, and satisfies 
\begin{subequations}
\begin{align}
\label{eqn:DivPropsa} &\Div_{\p} (f \,X) =  f\, \Div_{\p} X + (-1)^{\widetilde{f} \, \widetilde{X}}\, X(f)\,, \\
\label{eqn:DivPropsb} &\quad \Div_{\p'}X = \Div_{\p}X + X(g)\,,\\ 
\label{eqn:DivPropsc}&\Div_{\p}[X,Y] = X \Div_{\p} Y - (-1)^{\widetilde{X}\, \widetilde{Y}} Y \Div_{\p}X\,,
\end{align}
\end{subequations}
To clarify the notation used, for all $X,Y, \in \Vect(M)$, $f \in C^\infty(M)$,  and $\p' =\rme^g \, \p$ with $g \in C^\infty(M)$ even. For local expressions, the reader may consult \cite{Bruce:2017}. 
\begin{definition}
 An \emph{$\mathcal{N} =2$ supermanifold with a volume} is a quintuple $(M, Q_i, P, \p)$, where $M$ is a supermanifold,  $(Q_1, Q_2, P)$ is a set of vector fields that satisfy the  $\mathcal{N}=2$ $d=1$ supertranslation algebra \eqref{eqn:SupLieAlg}, and $\p$ is a Berezin volume.
\end{definition}
On any $\mathcal{N} =2$ supermanifold with a volume, there is a pair of standard cohomologies on functions (see \cite{Lyakhovich:2008,Lyakhovich:2010})
$$\Hom_{\mathrm{st}}(M, Q_i) := \frac{\{Q_i-\textnormal{closed functions}\}}{\{Q_i-\textnormal{exact functions}\}}\,,$$
and so a pair of modular classes 
$$\mathsf{Mod}(M, Q_i) := [\Div_{\p} Q_i]_{\mathrm{st}} \in\Hom_{\mathrm{st}}(M, Q_i) \,.$$
As $\{Q_1, Q_2\}=P \neq 0$, we do not have a bicomplex; or in supergeometric terms, we do not have a double Q-manifold\footnote{By definition, a double Q-manifold is a supermanifold equipped with two homological vector fields that commute.}. Furthermore, the modular classes $\mathsf{Mod}(M, Q_i)$ can easily be shown to be independent of the choice of Berezin volume using \eqref{eqn:DivPropsa}.  This directly implies that one can work locally using coordinates when evaluating the modular classes. It is essential to note that the modular class of a Q-manifold is a global invariant that can be meaningfully captured locally; this fact makes the modular class particularly convenient. \par
The vanishing of the modular class $\mathsf{Mod}(M, Q_i)$ is a necessary and sufficient condition for the existence of a Berezin volume that is invariant under the action of the supercharge $Q_i$. In particular, the divergence of the supercharge $Q_i$ vanishes for such a Berezin volume.  However, the vanishing of both modular classes is a necessary, but \ul{not} a sufficient condition for the existence of a Berezin volume that is invariant under both supercharges.
\begin{theorem}[Cohomological Coherence Criterion]\label{thm:BothInv}
Let $(M, Q_i, P, \p)$ be a $\mathcal{N} =2$ supermanifold with a volume. Assume that both modular classes $\mathsf{Mod}(M, Q_i)$ vanish, i.e., there exists Berezin volumes $\p_i$, such that $\Div_{\p_i}Q_i =0$. Then there exists a Berezin volume $\p'$ such that  $\Div_{\p'}Q_i =0$ ($i = 1,2$) if and only if  there exist an even function $\psi \in C^\infty(M)$ such that
$$\Div_{\p_1} Q_2 = Q_2(\psi)\,, \quad \textnormal{and} \qquad Q_1(\psi) =0\,.$$
\end{theorem}
\begin{proof}\
\begin{description}
    \item[If]  Assuming that $\Div_{\p_1} Q_2 = Q_2(\psi)$ and $Q_1(\psi)=0$ for some even function $\psi \in C^\infty(M)$. We then claim $\p':= \rme^{- \psi}\, \p{_1}$ is the required Berezin volume.  By construction, we have 
    $$\Div_{\p'} Q_2= \Div_{\p_1} Q_2 - Q_2(\psi) =0\,.$$
    Similarly, as $\Div_{\p_1} Q_1 =0$, we have
    $$\Div_{\p'} Q_1= \Div_{\p_1} Q_1 - Q_1(\psi) = - Q_1(\psi) =0\,.$$
    \item[Only if] Assuming that such a Berezin volume $\p'$ exists, $\Div_{\p'}Q_1=0$ implies that $\p_1 =  \rme^{\psi}\, \p'$ for some even $\psi \in C^\infty(M)$. As $\Div_{\p_1}Q_1 =0$,
    $$\Div_{\p'}Q_1 =  \Div_{\p_1}Q_1 - Q_1(\psi) =- Q_1(\psi)=0\,,$$
    and so $Q_1(\psi)=0$. As $\Div_{\p'} Q_2 =0$, we have
    $$\Div_{\p_1} Q_2 =  \Div_{\rme^\psi\,\p'}Q_2 = \Div_{\p'} Q_2 + Q_2(\psi) = Q_2(\psi)\,.$$
\end{description}
\end{proof}
Theorem \ref{thm:BothInv} is symmetric under the exchange of $1 \leftrightarrow 2$ due to the form of the Lie superalgebra \eqref{eqn:SupLieAlg}.  In particular, if we have 
\begin{align*}
\Div_{\p} Q_1 = Q_1(\psi_1)\,,& \qquad Q_2(\psi_1) =0\,,&\textnormal{and}&&
\Div_{\p} Q_2 = Q_2(\psi_2)\,,& \qquad Q_1(\psi_2) =0\,,
\end{align*}
then $\p' = \rme^{-(\psi_1 + \psi_2)}\, \p$ provides the (generally non-unique) Berezin volume  of Theorem \ref{thm:BothInv}. Clearly, $\Div_{\p'}Q_i =0$ directly implies $L_{Q_i} \p' =0$  and so the Berezin volume $\p'$ is `supersymmetric'. Locally, using coordinates and assuming the conditions of the theorem, we can construct $\p'$ by using the coordinate Berezin volume and explicitly find suitable $\psi_1$ and $\psi_2$. \par 
If the conditions of Theorem \ref{thm:BothInv} hold, then 
$$\Div_{\p'} Q_\pm = \Div_{\p'} Q_1 \pm \Div_{\p'}Q_2 = 0\,,$$
In summary and recalling the algebra \eqref{eqn:SupLieAlg2}, we have the following corollary.
\begin{corollary}
Under the conditions of Theorem  \ref{thm:BothInv}, there exists a Berezin volume $\p'$ that is invariant under the action of the supercharges $Q_\pm$, i.e., $L_{Q_\pm} \p' =0$.
\end{corollary}
The converse holds, in that if there is a Berezin volume that is invariant under the action of both $Q_\pm$, then the Berezin volume is also invariant under both $Q_1$ and $Q_2$. \par 
From the properties of the divergence operator \eqref{eqn:DivPropsc} we have
$$\Div_{\p}P = \Div_{\p}\{Q_1 , Q_2\} =Q_1(\Div_{\p} Q_2) + Q_2(\Div_{\p} Q_1)\,. $$\par
Thus, if there exists a Berezin volume $\p'$ that is invariant under both $Q_1$ and $Q_2$, then  $\Div_{\p'}P =0$, and so $L_P \p' =0$. \par 
In general, the volume is not preserved along $P$; however, the modular classes are well-behaved with respect to the action of $P$. More carefully, from \eqref{eqn:SupLieAlg}, we have $[P, Q_i] =0$. Thus,
$$\Div_{\p}[P, Q_i] = P(\Div_{\p}Q_i) - Q_i (\Div_{\p}P) =0\,.$$
Hence, $P(\Div_{\p}Q_i) = Q_i (\Div_{\p}P)$, which is $Q_i$-exact and so by definition, is zero in the respective standard cohomology. This implies that $P\big([\Div_{\p} Q_i]_{\mathrm{st}}\big)$ is zero in the standard cohomology of $Q_i$, respectively. Thus, the modular classes are invariant under the action of the even generator $P$.\par
We conclude this letter with simple illustrative examples. Recall that in (local) coordinates $x^a$ on a supermanifold, the divergence of a vector field is given by 
$$\Div_{\p} X = (-1)^{\widetilde{a}\,(\widetilde{X} +1)}\frac{1}{\rho}\frac{\partial(X^a\rho)}{\partial x^a}\,,$$
where (locally) $\p = D[x]\rho(x)$. Here $D[x]$ is the coordinate Berezin volume and $\rho$ is even and nowhere vanishing.  The Grassmann parity we denote as $\widetilde{x^a} := \widetilde{a} \in \Z_2$, and $\widetilde{X} \in \Z_2$.  For brevity, we will set $\mathsf{Mod}(Q_i) := \mathsf{Mod}(M, Q_i)$ as the underlying supermanifold will be clear.
\begin{example}
Consider $\R^{1|2}$ equipped with global coordinates $(t, \theta^1, \theta^2)$ and vector fields
$$Q_1 = \frac{\partial}{\partial \theta^1} + \frac{1}{2} \theta^2 \frac{\partial}{\partial t}\,,\qquad Q_2 = \frac{\partial}{\partial \theta^2} + \frac{1}{2} \theta^1 \frac{\partial}{\partial t}\,, \qquad P = \frac{\partial}{\partial t}\,.$$
The reader can quickly confirm that these vector fields satisfy \eqref{eqn:SupLieAlg}. As the modular classes are independent of the chosen Berezin volume, we can pick the coordinate Berezin volume $\p = D[t, \theta^2, \theta^1]$. Then 
\begin{align*}
& \Div_{\p} Q_1 =  \frac{\partial}{\partial\theta^1}(1) + \frac{\partial}{\partial\theta^2}(0) + \frac{\partial}{\partial t}\big(  \half \theta^2\big) =0\,, 
&\Div_{\p} Q_2 =  \frac{\partial}{\partial\theta^1}(0) + \frac{\partial}{\partial\theta^2}(1) + \frac{\partial}{\partial t}\big(  \half \theta^1\big) =0\,.
\end{align*}
Thus, both $\mathsf{Mod}(Q_i)$ vanish and so there exists an invariant Berezin volume; in this case, the coordinate Berezin volume is invariant.
\end{example}
\begin{example}
Consider $\R^{1|2}$ equipped with global coordinates $(t, \theta^1, \theta^2)$   and the vector fields
$$Q_1 = \frac{\partial}{\partial \theta^1} + \mu(t) \theta^2 \frac{\partial}{\partial t}\,, \qquad Q_2 = \frac{\partial}{\partial \theta^2}\,,\qquad P = \mu(t)\frac{\partial}{\partial t}\,.$$
As the modular classes are independent of the chosen Berezin volume, we can pick the coordinate volume $\p = D[t, \theta^1 , \theta^2]$. Then 
\begin{align*}
    & \Div_{\p} Q_1 = \frac{\partial}{\partial \theta^1}(1) + \frac{\partial}{\partial \theta^2}(0) + \frac{\partial}{\partial t}(\mu(t) \theta^2) = \mu'(t)\theta^2 = Q_1\big(\log(\mu(t))\big)\,, \\
    & \Div_{\p} Q_2=\frac{\partial}{\partial \theta^1}(0) + \frac{\partial}{\partial \theta^2}(1) + \frac{\partial}{\partial t}(0) =0\,,
\end{align*}
where we have taken $\mu(t) > 0$ for all $t \in \R$. Thus, both the modular classes vanish, i.e., $\mathsf{Mod}(Q_i)=0$. Moreover, we observe that $Q_2(\log(\mu(t))) =0$ and so we can find a Berezin volume that is invariant under both supercharges. Specifically, $\p' = \rme^{- \log(\mu(t))}\, \p$ is such a Berezin volume.
\end{example}  
\begin{example}\label{exa:R22}
Consider $\R^{2|2}_{>0}$ equipped with global coordinates $(x,y, \theta^1, \theta^2)$,  by definition we have $x,y > 0$. We equip this supermanifold with the vector fields
$$Q_1 = \frac{\partial}{\partial \theta^1} + \theta^2 x \frac{\partial}{\partial x}\,, \qquad Q_2 = \frac{\partial}{\partial \theta^2} + \theta^1 y \frac{\partial}{\partial y}\,, \qquad P  = x \frac{\partial}{\partial x} +y \frac{\partial}{\partial y}\,.$$
The reader can quickly confirm that these vector fields satisfy \eqref{eqn:SupLieAlg}. As the modular classes are independent of the chosen Berezin volume, we can pick the coordinate Berezin volume $\p = D[x,y, \theta^1, \theta^2]$. Then
\begin{align*}
& \Div_{\p}Q_1 =  \frac{\partial}{\partial \theta^1}(1) + \frac{\partial}{\partial \theta^2}(0) + \frac{\partial}{\partial x}(\theta^2 x) + \frac{\partial}{\partial y}(0)   = \theta^2\,,  \\
& \Div_{\p}Q_2 =  \frac{\partial}{\partial \theta^1}(0) + \frac{\partial}{\partial \theta^2}(1) + \frac{\partial}{\partial x}(0) + \frac{\partial}{\partial y}(\theta^1 y)   = \theta^1\,.
\end{align*}
Note that
\begin{align*}
&\theta^2 = Q_1(a\,\theta^1 \theta^2 + b\, \ln x + f(y))  \implies  \mathsf{Mod}(Q_1) =0\,,\\
&\theta^1 = Q_2(c\,\theta^2 \theta^1 + d\, \ln y +g(x)) \implies  \mathsf{Mod}(Q_2) =0\,,
\end{align*}
with $a,b,c,d \in R$ such that $a+b = 1$, and $c+d =1$.  There exists invariant Berezin volumes $\p_i$ such that $\Div_{\p_i}Q_i =0$. ($i= 1,2$). We explicitly see that $\p_1 =  \rme^{-(a \, \theta^1 \theta^2 + b \ln x + f(y))\theta^1 \theta^2} \, \p$, where $\p$ is the coordinate Berezin volume. Checking this, we have 
\begin{align*}
    & \Div_{\p_1}Q_1 =  \theta ^2 - Q_1( a \, \theta^1 \theta^2 + b \ln x + f(y)) = \theta^2 - a \, \theta^2 - b\, \theta^2 =0\,,\\
    & \Div_{\p_1} Q_2 = \theta^1 + Q_2 (a \, \theta^1 \theta^2 + b \ln x + f(y))=  \theta^1 +\theta^1 =\theta^1 + a\, \theta^1 - \theta^1 \, y  f'(y)\,.
\end{align*}
Setting $\psi := (1+a) \theta^2 \theta^1 - f(y)$, we observe that $\Div_{\p_1} Q_2 = Q_2(\psi)$. However, $Q_1(\psi) =  -(1+a)\theta^2$, and so $\psi$  is $Q_1$-closed if and only if $a= -1$, which implies that $b= 2$. Thus, in light of Theorem \ref{thm:BothInv}, we know there exists a Berezin volume that invariant under the action of both the supercharges. We modify $\p_1$ using $\psi$ by gives 
$$\p' = \rme^{\theta^1 \theta^2 - 2 \ln x}\p\,,$$
as an invariant Berezin volume. However, this is not unique  and noting the symmetry of the vector fields under $x \leftrightarrow y$ and $\theta^1 \leftrightarrow \theta^2$ tells us that 
$$\p'' = \rme^{\theta^2 \theta^1 - 2 \ln y}\p\,,$$
is another invariant Berezin volume.
\end{example}
\begin{example}
Consider $\R^{2|2}$ equipped with global coordinates $(x,y, \theta^1, \theta^2)$ and vector fields
$$Q_1 = \frac{\partial}{\partial \theta^1} + \theta^2 x \frac{\partial}{\partial x}\,, \qquad Q_2 = \frac{\partial}{\partial \theta^2} + \theta^1 y \frac{\partial}{\partial y}\,, \qquad P  = x \frac{\partial}{\partial x} +y \frac{\partial}{\partial y}\,.$$
 The reader can quickly confirm that these vector fields satisfy \eqref{eqn:SupLieAlg}. As the modular classes are independent of the chosen Berezin volume, we can pick the coordinate Berezin volume $\p = D[x,y, \theta^1, \theta^2]$. Then
\begin{align*}
& \Div_{\p}Q_1 =  \frac{\partial}{\partial \theta^1}(1) + \frac{\partial}{\partial \theta^2}(0) + \frac{\partial}{\partial x}(\theta^2 x) + \frac{\partial}{\partial y}(0)   = \theta^2\,,  \\
& \Div_{\p}Q_2 =  \frac{\partial}{\partial \theta^1}(0) + \frac{\partial}{\partial \theta^2}(1) + \frac{\partial}{\partial x}(0) + \frac{\partial}{\partial y}(\theta^1 y)   = \theta^1\,.
\end{align*}
Note that
\begin{align*}
&\theta^2 = Q_1(\theta^1 \theta^2 +f(y)) \implies  \mathsf{Mod}(Q_1) =0\,,
&&\theta^1 = Q_2(\theta^2 \theta^1+ g(x)) \implies \mathsf{Mod}(Q_2) =0\,.
\end{align*}
Note that, unlike Example \ref{exa:R22}, we do not have the terms $\ln x$ and $\ln y$, as these are not smooth functions functions on $\R$.  Thus, there exists invariant Berezin volumes $\p_i$ such that $\Div_{\p_i}Q_i =0$. ($i= 1,2$). We explicitly see that $\p_1 =  \rme^{-(\theta^1 \theta^2 + f(y)} \, \p$, where $\p$ is the coordinate Berezin volume. Checking this, we have 
\begin{align*}
    & \Div_{\p_1}Q_1 =  \theta ^2 - Q_1( \theta^1 \theta^2 + f(y) )= \theta^2 - \theta^2 =0\,,\\
    & \Div_{\p_1} Q_2 = \theta^1 - Q_2 (\theta^1 \theta^2 + f(y))=  \theta^1 +\theta^1 = 2 \, \theta^1 - \theta^1 \, y f'(y)\,.
\end{align*}
Setting $\psi := 2 \theta^2 \theta^1- f(y)$, we observe that $\Div_{\p_1} Q_2 = Q_2(\psi)$. However, $Q_1(\psi) =  -2\, \theta^2 \neq 0$, and so $\psi$  is not $Q_1$-closed. Thus, in light of Theorem \ref{thm:BothInv}, we know there cannot be a single Berezin volume that is invariant under both $Q_1$ and $Q_2$.
\end{example}
%
%
\section*{Acknowledgements}
The author cordially thanks Francesco Toppan for comments on earlier drafts of this note.
%
%

\end{document}